\documentclass{article}

\usepackage{microtype}
\usepackage{graphicx}
\usepackage{subfigure}
\usepackage{booktabs} 

\usepackage{hyperref}

\usepackage[accepted]{icml2018}

\usepackage{amsmath}
\usepackage{amssymb}
\usepackage{amsthm}

\newtheorem{theorem}{Theorem}
\newtheorem{lemma}{Lemma}

\icmltitlerunning{Topological Mixture Estimation}

\begin{document}

\twocolumn[
\icmltitle{Topological Mixture Estimation}



\icmlsetsymbol{equal}{*}

\begin{icmlauthorlist}
\icmlauthor{Steve Huntsman}{bae}
\end{icmlauthorlist}

\icmlaffiliation{bae}{BAE Systems FAST Labs, Arlington, VA, USA}

\icmlcorrespondingauthor{Steve Huntsman}{steve.huntsman@baesystems.com}

\icmlkeywords{Machine Learning, ICML}

\vskip 0.3in
]



\printAffiliationsAndNotice{}  

\begin{abstract}
We introduce \emph{topological mixture estimation}, a completely nonparametric and computationally efficient solution to the problem of estimating a one-dimensional mixture with generic unimodal components. 
We repeatedly perturb the unimodal decomposition of Baryshnikov and Ghrist to produce a topologically and information-theoretically optimal unimodal mixture. 
We also detail a smoothing process that optimally exploits topological persistence of the unimodal category in a natural way when working directly with sample data. 
Finally, we illustrate these techniques through examples.
\end{abstract}

\section{\label{sec:Introduction}Introduction}

\subsection{\label{sec:background}Background}

Density functions that represent sample data are often multimodal, i.e. they exhibit more than one maximum. Typically this behavior indicates that the underlying data deserves a more detailed representation as a mixture of densities with individually simpler structure. The usual specification of a component density is quite restrictive, with log-concave the most general case considered in the literature, and Gaussian the overwhelmingly typical case. It is also necessary to determine the number of mixture components \emph{a priori}, and much art is devoted to this. 

In this paper we detail how to efficiently determine a topologically and information-theoretically optimal mixture of generic unimodal component densities directly from a one-dimensional input density and without any auxiliary information whatsoever. The topological criterion is a natural qualitative alternative to more traditional quantitative model selection criteria (e.g., information criteria) and is computed at the outset of computation, then subsequently preserved, while the information-theoretical criterion optimally separates component densities. We further show how to optimally smooth the mixture when the input density itself is being estimated. Topological persistence (which operationally amounts to the assignment of significance to topological features that persist as a function of scale) is the essential ingredient in both the ``model selection'' and smoothing.

\subsection{\label{sec:formalMotivation}Formal Motivation}

To give some formal motivation, let $\mathcal{D}(\mathbb{R}^d)$ denote a suitable space of continuous probability densities (henceforth merely called densities) on $\mathbb{R}^d$. A \emph{mixture} on $\mathbb{R}^d$ with $M$ components is a pair $(\pi,p) \in \Delta_M^\circ \times \mathcal{D}(\mathbb{R}^d)^M$, where $\Delta_M^\circ := \{\pi \in (0,1]^M : \sum_m \pi_m = 1\}$; we write $|(\pi,p)| := M$, and note that $\pi$ cannot have any components equal to zero. The corresponding \emph{mixture density} is $\langle \pi, p \rangle := \sum_{m=1}^M \pi_m p_m$. The \emph{Jensen-Shannon divergence} of $(\pi,p)$ is \cite{briet2009properties}
\begin{equation}
\label{eq:JS}
J(\pi, p) := H \left ( \langle \pi, p \rangle \right ) - \langle \pi, H(p) \rangle
\end{equation}
where $H(p)_m := H(p_m)$ and $H(f) := -\int f \log f \ dx$ is the entropy of $f$. 

Now $J(\pi,p)$ is the mutual information between the random variables $\Xi \sim \pi$ and $X \sim \langle \pi, p \rangle$. Since mutual information is always nonnegative, the same is true of $J$. The concavity of $H$ gives the same result, i.e. $H \left ( \langle \pi, p \rangle \right ) \ge \langle \pi, H(p) \rangle$. If $M := |(\pi,p)| > 1$, $\hat \pi := \left ( \pi_1,\dots,\pi_{M-2},\pi_{M-1}+\pi_M \right )$, and $\hat p := \left ( p_1,\dots,p_{M-2},\frac{\pi_{M-1} p_{M-1}+\pi_M p_M}{\pi_{M-1}+\pi_M} \right )$, then is easy to show that $J(\hat \pi, \hat p) \le J(\pi, p)$.

We say that a density $f \in \mathcal{D}(\mathbb{R}^d)$ is \emph{unimodal} if $f^{-1}([y,\infty))$ is either empty or contractible (i.e., topologically equivalent to a point in the sense of homotopy) for all $y$. For $d=1$, this simply means that any nonempty sets $f^{-1}([y,\infty))$ are intervals and agrees with intuition. We call a mixture $(\pi, p)$ unimodal iff each of the component densities $p_m$ is unimodal. The \emph{unimodal category} $\text{ucat}(f)$ is the smallest number of components of any unimodal mixture $(\pi,p)$ that satisfies $\langle \pi, p \rangle = f$. Figure \ref{fig:unidec} shows that the unimodal category can be much less than the number of maxima. In the event that $\langle \pi, p \rangle = f$ and $|(\pi,p)| = \text{ucat}(f)$, we write $(\pi, p) \models f$: the symbol $\models$ is called ``models.'' The unimodal category is a topological invariant that generalizes and relates to other classical invariants \cite{baryshnikov2011unimodal, ghrist2014elementary}.

\begin{figure}[htbp]
\includegraphics[trim = 10mm 80mm 10mm 11mm, clip, width=\columnwidth, keepaspectratio]{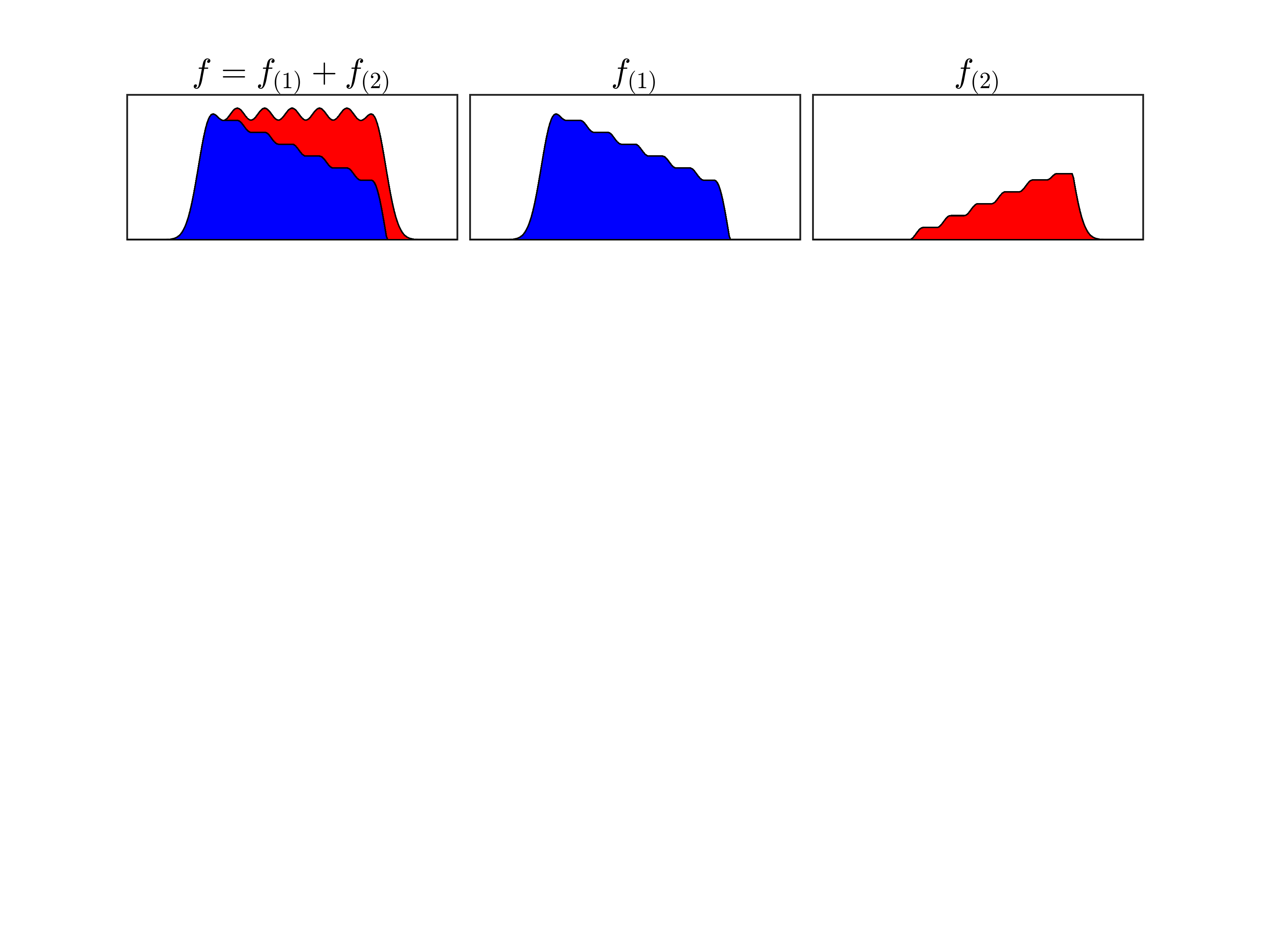}
\\
\includegraphics[trim = 10mm 80mm 10mm 11mm, clip, width=\columnwidth, keepaspectratio]{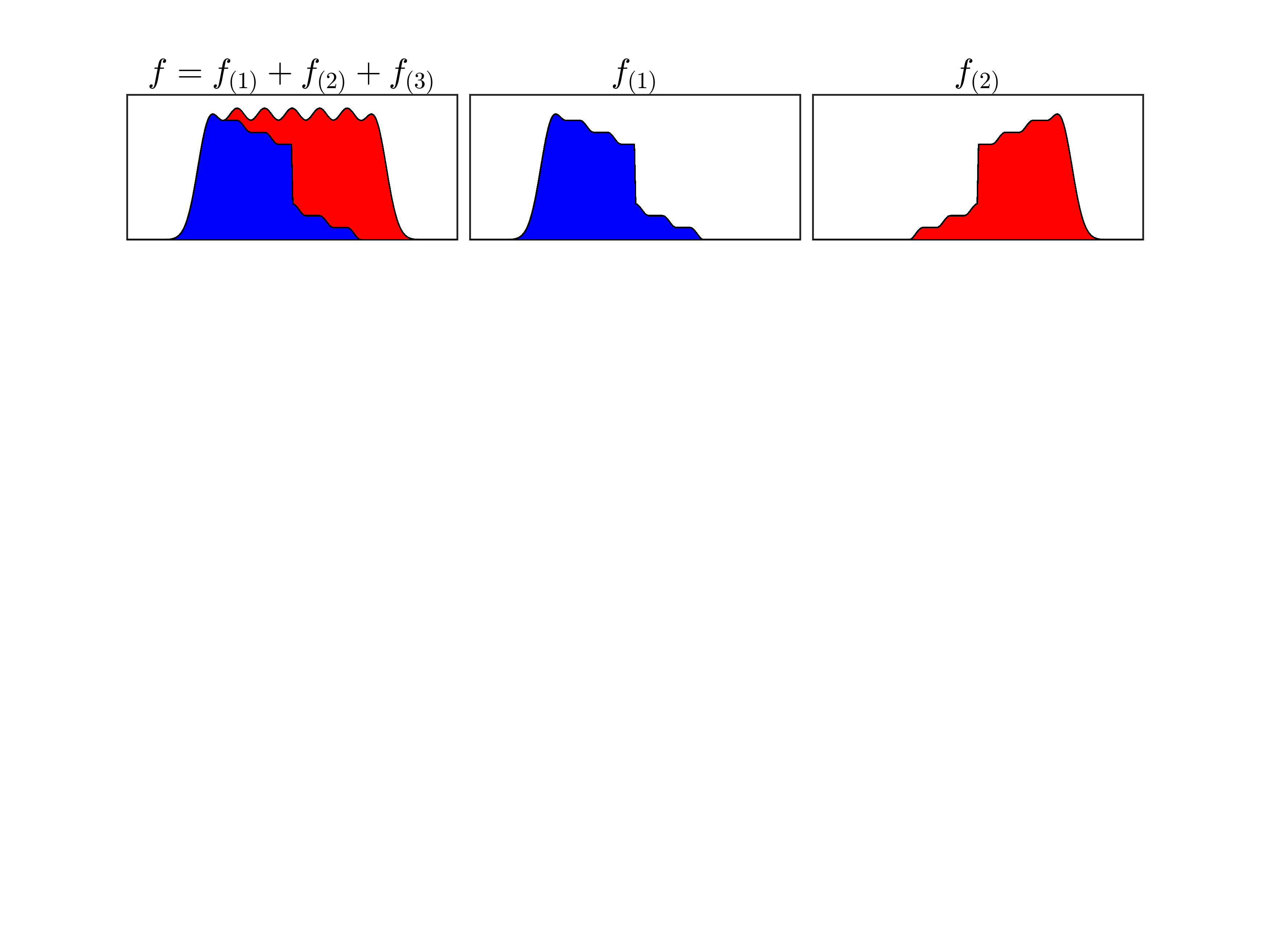}
\caption{ \label{fig:unidec} Upper panels: a unimodal decomposition obtained using the ``sweep'' algorithm from \cite{baryshnikov2011unimodal}. Lower panels: the result of \eqref{eq:TME}. The de/reblurring approach of \S \ref{sec:deblurring} gives smoother decompositions on estimated densities: see Figures \ref{fig:OldFaithfulArea}-\ref{fig:ColorIndicesLine}.
}
\end{figure} %

The preceding constructions naturally lead us to consider the \emph{unimodal Jensen-Shannon} divergence
\begin{equation}
\label{eq:UJS}
J_\cap(f) := \sup_{(\pi, p) \models f} J(\pi, p)
\end{equation}
as a simultaneous measure of both the topological and information-theoretical complexity of $f$, and
\begin{equation}
\label{eq:TME}
(\pi_\cap, p_\cap) := \arg \max_{(\pi, p) \models f} J(\pi, p)
\end{equation}
as a topologically and information-theoretically optimal \emph{topological mixture estimate} (TME). 

The natural questions are if such an estimate exists (is the supremum attained?), is unique, and if so, how to perform TME in practice. In this paper we address these questions for the case $d = 1$, and we demonstrate the utility of TME in examples (see Figures \ref{fig:OldFaithfulArea}-\ref{fig:CvLpmodeTmeBumpHunting}).

After reviewing related work in \S \ref{sec:Related}, we cover the basic algorithm of TME in \S \ref{sec:Algorithm}. The proof therein that Algorithm \ref{alg:tme} computes \eqref{eq:TME} is nearly trivial with Lemmas \ref{lem:convexity} and \ref{lem:unimodality} in hand: these are respectively in appendices \S \ref{sec:Convexity} and \S \ref{sec:PreservingUnimodality}. Next, in \S \ref{sec:TDE} we review the related technique of \emph{topological density estimation} (TDE) before showing in \S \ref{sec:blurringDeblurring} how blurring and deblurring mixture estimates can usefully couple TDE and TME. Finally, in \S \ref{sec:examples} we produce examples of TME in action before making some closing remarks in \S \ref{sec:remarks}.

\section{\label{sec:Related}Related Work}

While density estimation enables various clustering techniques \cite{li2007nonparametric, JSSv057i11,xu2015clustering}, mixture estimation is altogether more powerful than clustering: e.g., it is possible to have mixture components that significantly and meaningfully overlap. For example, a cluster with a bimodal density will usually be considered as arising from two unimodal mixture components that are individually of interest. In this light and in view of its totally nonparametric nature, our approach can be seen as particularly powerful, particularly when coupled with TDE and deblurring/reblurring (see \S \ref{sec:TDE} and \S \ref{sec:blurringDeblurring}).

Still, even for clustering (even in one dimension, where an optimal solution to $k$-means can be computed efficiently \cite{wang2011ckmeans, nielsen2014optimal, gronlund2017fast}), determining the number of clusters in data \cite{feng2007pgmeans,mirkin2011choosing} is as much an art as a science. All of the techniques we are aware of either require some \emph{ad hoc} determination to be made, require auxiliary information (e.g., \cite{tibshirani2001estimating}) or are parametric in at least a limited sense (e.g., \cite{sugar2003finding}). While a parametric approach allows likelihoods and thus various information criteria \cite{burnham2003model} or their ilk to be computed for automatically determining the number of clusters, this comes at the cost of a strong modeling assumption, and criteria values themselves are difficult to compare meaningfully \cite{melnykov2010finite}.

These shortcomings--including determining the number of mixture components--carry over to the more difficult problem of mixture estimation. \cite{mclachlan2004finite, melnykov2010finite, mclachlan2014number} As an example, an \emph{ad hoc} and empirically derived unimodal mixture estimation technique that requires one of a few common functional forms for the mixture components has been recently employed in \cite{mints2017unified}. Univariate model-based mixtures of skew distributions admit EM-type algorithms and can outperform Gaussian mixture models \cite{lin2007robust, basso2010robust}. Though these generalize to the multivariate case quite effectively (see, e.g., \cite{lee2016finite}), the EM-type algorithms are generically vulnerable to becoming trapped in local minima without good initial parameter values, and they require some model selection criterion to determine the number of mixture components, though the parameter learning and model selection steps can be integrated as in \cite{figueiredo2002unsupervised}. A Bayesian nonparametric mixture model that incorporates many--but not arbitrary--unimodal distributions is considered in \cite{rodriguez2014univariate}. Principled work has been done on estimating mixtures of log-concave distributions \cite{walther2009inference} and \cite{chan2013learning} describes how densities of discrete unimodal mixtures can be estimated. However, actually estimating generic unimodal mixtures themselves appears to be unaddressed in the literature, even in one dimension. Indeed, even estimating individual modes and their associated uncertainties or significances has only been addressed recently \cite{genovese2016non, mukhopadhyay2017large}.

\section{\label{sec:Algorithm}The Basic Algorithm}

Given $f$, the ``sweep'' algorithm of \cite{baryshnikov2011unimodal} yields $(\pi, p) \models f$. We will repeatedly perturb $(\pi, p)$ to obtain \eqref{eq:TME} using Lemmas \ref{lem:convexity} and \ref{lem:unimodality}, which are respectively in \S \ref{sec:Convexity} and \S \ref{sec:PreservingUnimodality}. Lemma \ref{lem:convexity} states that that $J$ is convex under perturbations of $(\pi, p)$ that preserve $\langle \pi, p \rangle$. Lemma \ref{lem:unimodality} is a characterization of perturbations of two components of a piecewise affine and continuous (or piecewise constant) mixture that preserve the predicate $(\pi, p) \models f$, i.e., that preserve unimodality (as in Figure \ref{fig:unimodalGiveTake}) and the mixture density. Together, these results entail Theorem \ref{thm:tme}, which establishes that greedy unimodality- and density-preserving local perturbations of pairs of mixture components converge to \eqref{eq:TME}. 

\begin{figure}[htbp]
\includegraphics[trim = 22mm 10mm 15mm 5mm, clip, width = \columnwidth, keepaspectratio]{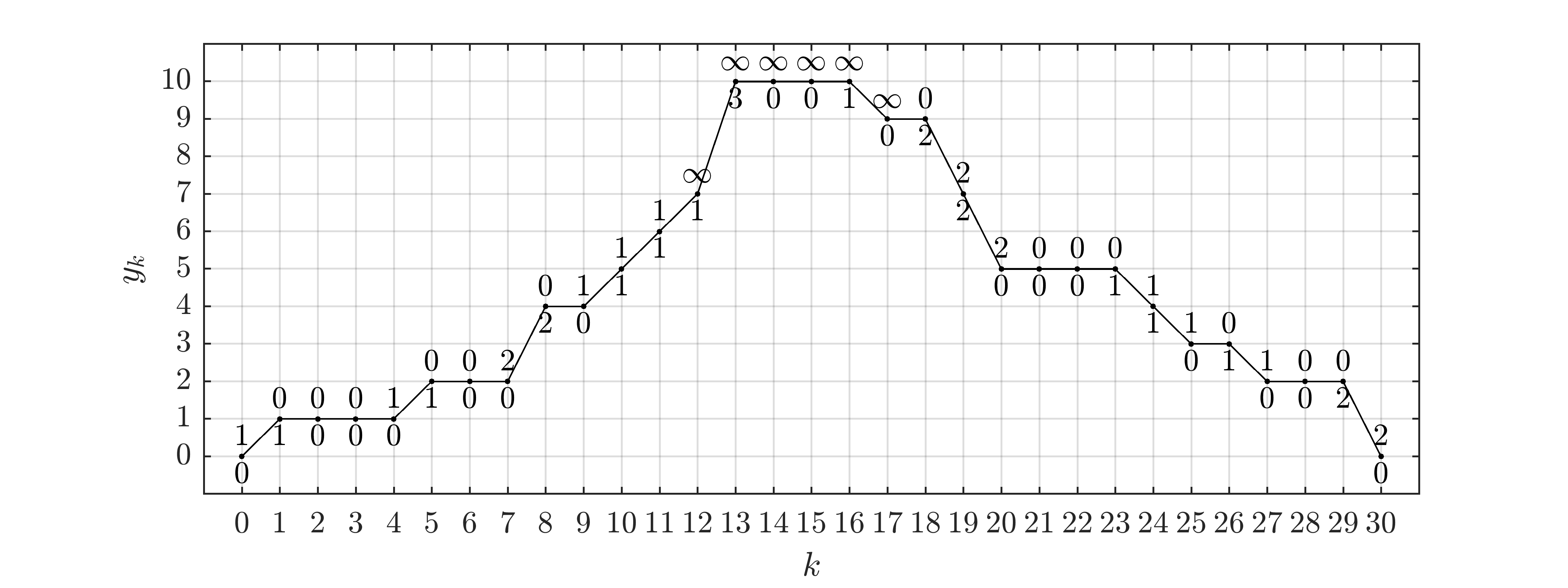}
\caption{ \label{fig:unimodalGiveTake} A unimodal sequence with unimodality-saturating local perturbations $\varepsilon^\pm$ from Lemma \ref{lem:unimodality} in \S \ref{sec:PreservingUnimodality} indicated above and below.
} 
\end{figure} %

\begin{theorem}
\label{thm:tme}
Let $-\infty = x_{-1} < x_0 < \dots < x_N < x_{N+1} = \infty$ and $f$ be piecewise constant (or affine) over each $[x_k,x_{k+1}]$. Then Algorithm \ref{alg:tme} efficiently computes \eqref{eq:TME}.
\end{theorem}

\begin{proof} 
By Lemma \ref{lem:convexity} (see \S \ref{sec:Convexity}), greedily and locally perturbing the mixture $(\pi, p) \models f$ according to Lemma \ref{lem:unimodality} (see \S \ref{sec:PreservingUnimodality}), then updating the mixture according to the perturbation which optimizes $J$ gives the desired result in $O(M N)$ iterations. This result is unique by convexity. Each iteration requires $O(M^2 N)$ trial perturbations, each of which in turn requires $O(M N)$ arithmetic operations to evaluate $J$.
\end{proof}

\begin{algorithm}[tb]
   \caption{Topological Mixture Estimation (TME)}
   \label{alg:tme}
\begin{algorithmic}
   \STATE {\bfseries Input:} function data $\{x_k,f(x_k)\}$
   \STATE Initialize $(\pi, p) \models f$ as in \cite{baryshnikov2011unimodal}
   \REPEAT
   \FOR{each evaluation point and pair of components}
   \STATE Greedily perturb $(\pi,p)$ to $(\pi',p')$ using Lemma \ref{lem:unimodality}
   \ENDFOR
   \STATE Update $(\pi,p) = \arg \max J(\pi',p')$
   \UNTIL{$(\pi,p)$ and/or $J(\pi,p)$ converge}
   \STATE {\bfseries Output:} $(\pi,p)$
\end{algorithmic}
\end{algorithm}

\section{\label{sec:TDE}Topological Density Estimation}

The obvious situation of practical interest for TME is that a density has been obtained from a preliminary estimation process involving some sample data. There is a natural approach to this preliminary estimation process called topological density estimation (TDE) \cite{huntsman2017topological} that naturally dovetails with TME.

\subsection{\label{sec:tdeIdea}Idea}

We recall the basic idea here (for pseudocode, see Algorithm \ref{alg:tde}). Given a kernel $K$ and sample data $X_j$ for $1 \le j \le n$, and for each proposed bandwidth $h$, compute the kernel density estimate \cite{silverman1986density, chen2017tutorial}
\begin{equation}
\label{eq:KDE}
\hat f_{h;X} := \frac{1}{n} \sum_{j=1}^n K_{X_j,h}
\end{equation}
where $K_{\mu,\sigma}(x) := \frac{1}{\sigma} K ( \frac{x-\mu}{\sigma} )$. Next, compute
\begin{equation}
\label{eq:ucatKDE}
u_X(h) := \text{ucat}(\hat f_{h;X})
\end{equation}
and estimate the unimodal category of the PDF that $X$ is sampled from via
\begin{equation}
\label{eq:ucatEstimate}
\hat m_X := \arg \max_m \mu(u_X^{-1}(m))
\end{equation}
where $\mu$ denotes an appropriate measure (nominally counting measure or the pushforward of Lebesgue measure under the transformation $h \mapsto 1/h$). 

\eqref{eq:ucatEstimate} gives the most prevalent and \emph{topologically persistent} \cite{ghrist2014elementary, oudot2015persistence} value of the unimodal category, i.e., this is a topologically robust estimate of the number of components required to produce the PDF that $X$ is sampled from as a mixture. While \emph{any} element of $u_X^{-1}(\hat m_X)$ is a bandwidth consistent with the estimate \eqref{eq:ucatEstimate}, considerations of robustness lead us to typically make the more detailed nominal specification
\begin{equation}
\label{eq:TDE}
\hat h_X := \text{median}_\mu ( u_X^{-1}(\hat m_X) ).
\end{equation}

\begin{algorithm}[tb]
   \caption{Topological Density Estimation (TDE)}
   \label{alg:tde}
\begin{algorithmic}
   \STATE {\bfseries Input:} $\{X_j\}$
   \FOR{each proposed bandwidth $h$}
   \STATE Compute $u_X(h)$ using \eqref{eq:ucatKDE}
   \ENDFOR
   \STATE Compute $\hat m_X$ using \eqref{eq:ucatEstimate}
   \STATE {\bfseries Output:} $\hat h_X$ using \eqref{eq:TDE}
\end{algorithmic}
\end{algorithm}

\subsection{\label{sec:tdePerformance}Performance}

TDE turns out to be very computationally efficient relative to the traditional technique of cross-validation (CV). On highly multimodal densities, TDE is competitive or at least reasonably performant relative to CV and other nonparametric density estimation approaches with respect to traditional statistical evaluation criteria. Moreover, TDE outperforms other approaches when \emph{qualitative} criteria such as the number of local maxima and the unimodal category itself are considered (see Figures \ref{fig:fkm}-\ref{fig:fkmEval500MatlabLMAX}). In practice, such qualitative criteria are generally of paramount importance. For example, precisely estimating the shape of a density is generally less important than determining if it has multiple modes. 

As an illustration, consider $\mu(j,m) := \frac{j}{m+1}$, $\sigma(k,m) := 2^{-(k+2)}(m+1)^{-2}$ and the family of distributions
\begin{equation}
\label{eq:fkm}
f_{km} := \frac{1}{m} \sum_{j = 1}^m K_{\mu(j,m),\sigma(k,m)}
\end{equation}
for $1 \le k \le 3$ and $1 \le m \le 10$, and where here $K$ is the standard Gaussian density: see Figure \ref{fig:fkm}. 
Exhaustive details relating to the evaluation of TDE on this family and other densities are in the software package and test suite \cite{BAETDE}: here, we merely show performance data for \eqref{eq:fkm} in Figures \ref{fig:fkmEval500MatlabUCAT} and \ref{fig:fkmEval500MatlabLMAX}.

\begin{figure}[htbp]
\includegraphics[trim = 30mm 95mm 30mm 20mm, clip, width=\columnwidth,keepaspectratio]{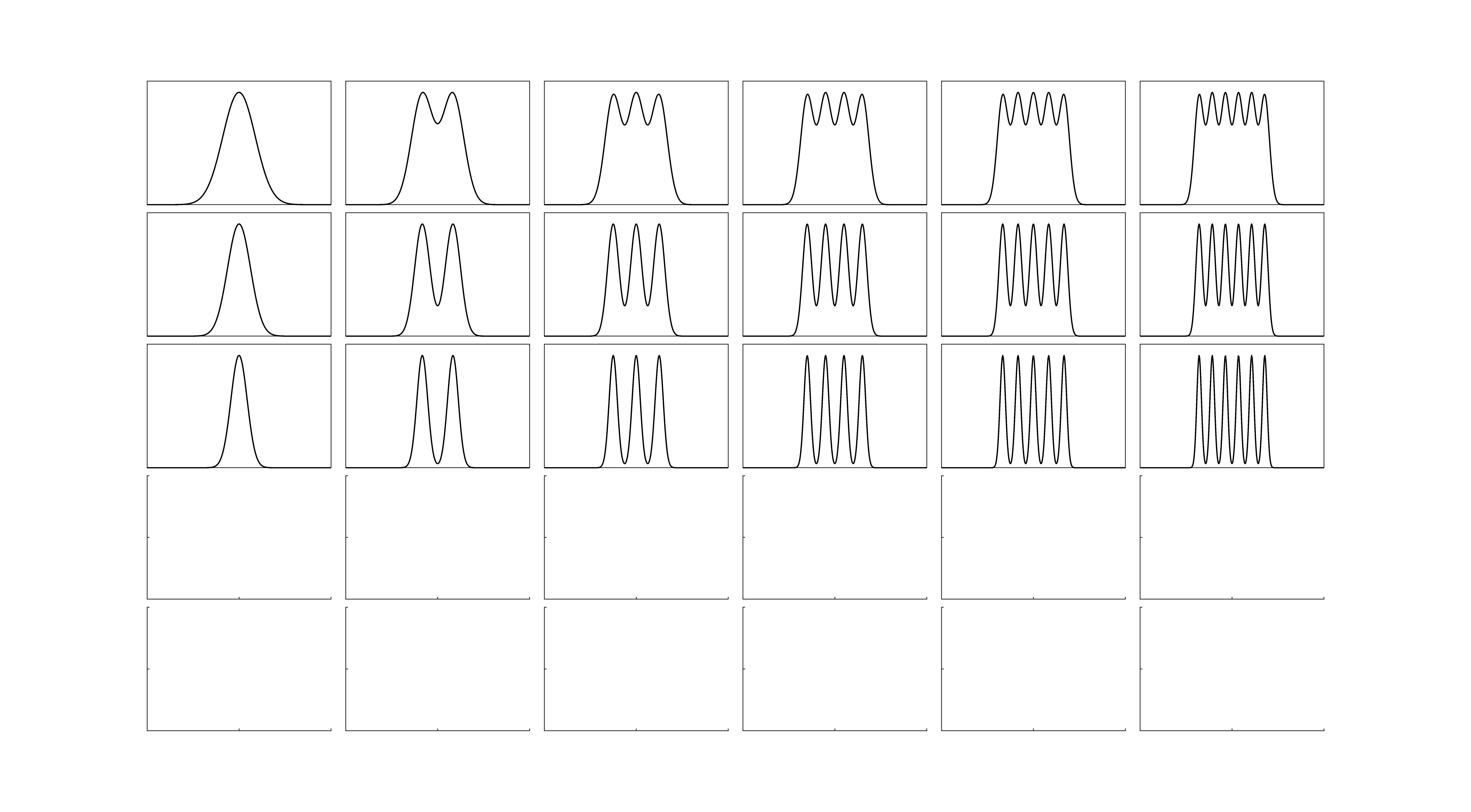}
\caption{ \label{fig:fkm} The densities $f_{km}$ in \eqref{eq:fkm} for $1 \le k \le 3$ and $1 \le m \le 6$ over $[-0.5,1.5]$. Rows are indexed by $k$; columns by $m$. The upper left panel shows $f_{11}$ and the lower right panel shows $f_{36}$. 
} 
\end{figure} %

\begin{figure}[htbp]
\includegraphics[trim = 5mm 2mm 10mm 3mm, clip, width=\columnwidth, keepaspectratio]{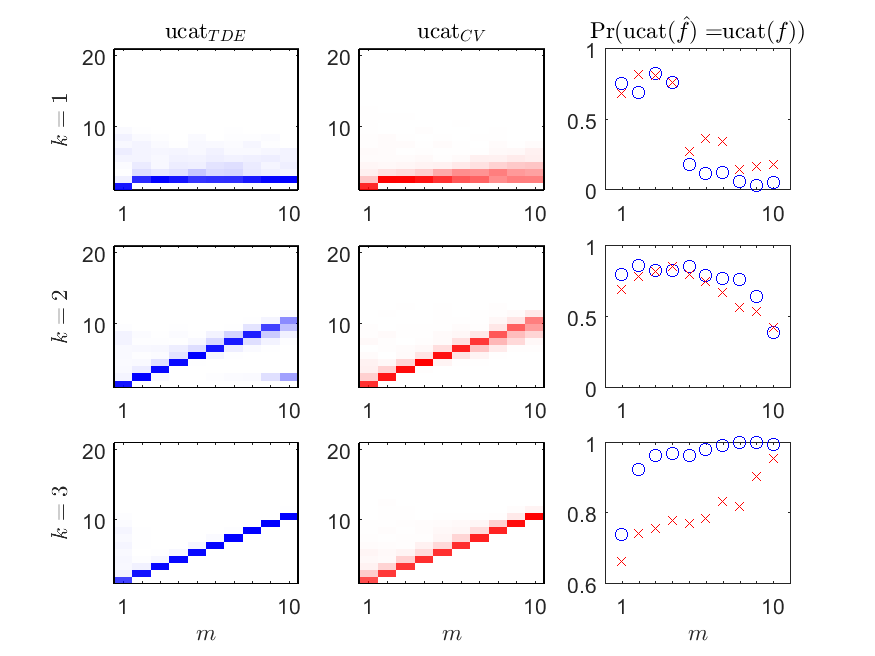}
\caption{ \label{fig:fkmEval500MatlabUCAT} Pseudotransparency plots of performance measures relating to the unimodal category for the family \eqref{eq:fkm} depicted in Figure \ref{fig:fkm} with $n = 500$ and using a Gaussian kernel. From left to right, we show empirical distributions of $\text{ucat}$ {\color{blue}(blue) for TDE}, $\text{ucat}$ {\color{red}(red) for CV}, and the empirical probability that the estimate of $\text{ucat}$ is correct. From top to bottom, we show $k = 1, \dots 3$. Each panel has $m = 1, \dots 10$ along the horizontal axis.}
\end{figure} %

\begin{figure}[htbp]
\includegraphics[trim = 5mm 2mm 10mm 3mm, clip, width=\columnwidth, keepaspectratio]{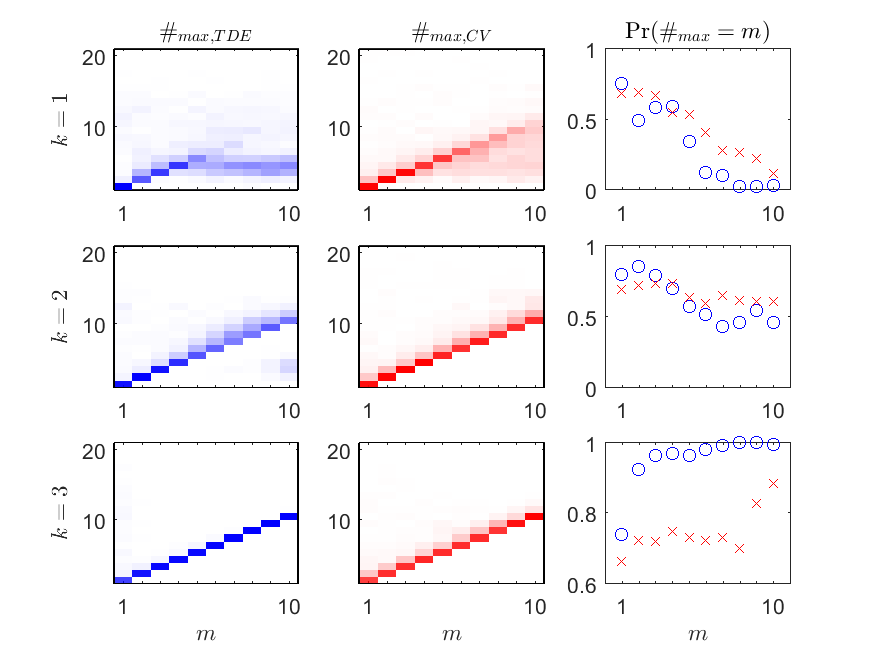}
\caption{ \label{fig:fkmEval500MatlabLMAX} As in Figure \ref{fig:fkmEval500MatlabUCAT}, but for the number of local maxima.}
\end{figure} %

TDE has the very useful feature (shared by essentially no high-performing density estimation technique other than CV) that it requires no free parameters or assumptions. Indeed, TDE can be used to evaluate its own suitability: for unimodal distributions, it is often not an ideal choice--but it is good at detecting this situation in the first place. Furthermore, TDE is very efficient computationally. 

In situtations of practical interest, it is tempting to couple TDE and TME in the obvious way: i.e., perform them sequentially and indepdently. This yields a completely nonparametric estimate of a mixture from sample data alone. However, there is a much better way to couple these techniques, as we shall see in the sequel.

\section{\label{sec:blurringDeblurring}Blurring and Deblurring}

\subsection{\label{sec:blurring}Blurring}

Recall that a log-concave function is unimodal, and moreover that a function is log-concave iff its convolutions with unimodal functions are identically unimodal \cite{ibragimov1956composition, keilson1971some, bertin2013unimodality}. This observation naturally leads to the following question: if $(\pi, p) \models f$, how good of an approximation to the $\delta$ distribution must a log-concave density $g$ be in order to have $(\pi, p * g) \models f * g$? In particular, suppose that $g$ is a Gaussian density: what bandwidth must it have? An answer to this question of how much blurring a minimal unimodal mixture model can sustain defines a topological scale (viz., the persistence of the unimodal category under blurring) that we proceed to illustrate in light of TDE.

In this paragraph we assume that $K$ is the standard Gaussian density, so that $K_{\mu,h} * K_{\mu',h'} = K_{\mu+\mu',(h^2 + h'^2)^{1/2}}$ and $\hat f_{h;X} * K_{0,h'} = \hat f_{(h^2 + h'^2)^{1/2};X}$. Write $\hat h_X$ for the bandwidth obtained via TDE, whether via the nominal specification \eqref{eq:TDE} or any other: by construction we have that $\inf u_X^{-1}(\hat m_X) \le \hat h_X \le \sup u_X^{-1}(\hat m_X)$. Now if $(\pi, p) \models \hat f_{\hat h_X;X}$, then $\hat m_X = u_X(\hat h_X) = |(\pi,p)|$. In order to have $(\pi, p * K_{0,h'}) \models f * K_{0,h'}$, it must be that $\hat m_X = u_X(\hat h_X) = |(\pi,p)| = |(\pi, p * K_{0,h'})| = u_X((\hat h_X^2+h'^2)^{1/2})$, i.e., 
\begin{equation}
\label{eq:blurring}
h' \le \left ( \left [ \sup u_X^{-1}(\hat m_X) \right ]^2 - \hat h_X^2 \right )^{1/2}.
\end{equation}
In particular, we have the weaker inequality involving a purely topological scale:
\begin{equation}
\label{eq:blurring2}
h' \le \left ( \left [ \sup u_X^{-1}(\hat m_X) \right ]^2 - \left [ \inf u_X^{-1}(\hat m_X) \right ]^2 \right )^{1/2}.
\end{equation}

The preceding considerations generalize straightforwardly if we define $u_f(h) := \text{ucat}(f * K_{0,h})$, where once again $K$ is a generic kernel. This generalizes \eqref{eq:ucatKDE} so long as we associate sample data with a uniform average of $\delta$ distributions. Under reasonable conditions, we can write $u_f(0) = \text{ucat}(f)$, and it is easy to see that the analogous bound is
\begin{equation}
\label{eq:blurringGeneral}
h' \le \sup u_f^{-1}(u_f(0)).
\end{equation}

Of course, \eqref{eq:blurringGeneral} merely restates the triviality that the blurred mixture ceases to be minimal precisely when the number of mixture components exceeds the unimodal category of the mixture density. Meanwhile, the special case furnished by TDE with the standard Gaussian kernel affords sufficient structure for a slightly less trivial statement.

\subsection{\label{sec:deblurring}Deblurring/reblurring}

The considerations of \S \ref{sec:blurring} suggest how to couple TDE and TME in a much more effective way than performing them sequentially and independently. The idea is to use a Gaussian kernel and instead of \eqref{eq:TDE}, pick the bandwidth
\begin{equation}
\label{eq:TDEminimal}
\hat h_X^{(-)} := \text{inf}_\mu ( u_X^{-1}(\hat m_X) )
\end{equation}
and then perform TME; finally, convolve the results with $K_{0,\Delta h}$ where
\begin{equation}
\label{eq:Deltah}
\Delta h := \left ( \hat h_X^2 - \left [ \hat h_X^{(-)} \right ]^2 \right )^{1/2}.
\end{equation}
This preserves the result of TDE while giving a smoother, less artificial, and more practically useful mixture estimate than the information theoretically optimal result.

Of course, a similar tactic can be performed directly on a density $f$ by considering its Fourier deconvolution $\mathcal{F}^{-1}(\mathcal{F}f/\mathcal{F}K_{0,h'})$, where $\mathcal{F}$ denotes the Fourier transform and $h'$ is as in \eqref{eq:blurringGeneral}: however, any \emph{a priori} justification for such a tactic is necessarily context-dependent in general, and our experience suggests that its implementation would be delicate and/or prone to aliasing. Nevertheless, this would be particularly desirable in the context of heavy-tailed distributions, where kernel density estimation requires much larger sample sizes in order to achieve acceptable results. In this context it would also be worth considering the use of a symmetric stable density \cite{uchaikin1999chance, nolan2018stable} (e.g., a Cauchy density) as a kernel with the aim of recapturing the essence of \eqref{eq:Deltah}.

\begin{algorithm}[tb]
   \caption{Reblurred Topological Mixture Estimation}
   \label{alg:blur}
\begin{algorithmic}
   \STATE {\bfseries Input:} $\{X_j\}$
   \FOR{each proposed bandwidth $h$}
   \STATE Compute $u_X(h)$ using \eqref{eq:ucatKDE}
   \ENDFOR
   \STATE Compute $\hat m_X$ using \eqref{eq:ucatEstimate}
   \STATE Compute $\hat h_X^{(-)}$ using \eqref{eq:TDEminimal}
   \STATE Compute $(\pi,p)$ from $\hat f_{\hat h_X^{(-)}; X}$ using Algorithm \ref{alg:tme}
   \STATE Update $p = p * K_{0,\Delta h}$ with $\Delta h$ as in \eqref{eq:Deltah}
   \STATE {\bfseries Output:} $(\pi,p)$
\end{algorithmic}
\end{algorithm}

\section{\label{sec:examples}Examples}

\begin{figure}[htbp]
\includegraphics[trim = 5mm 10mm 10mm 10mm, clip, width=\columnwidth,keepaspectratio]{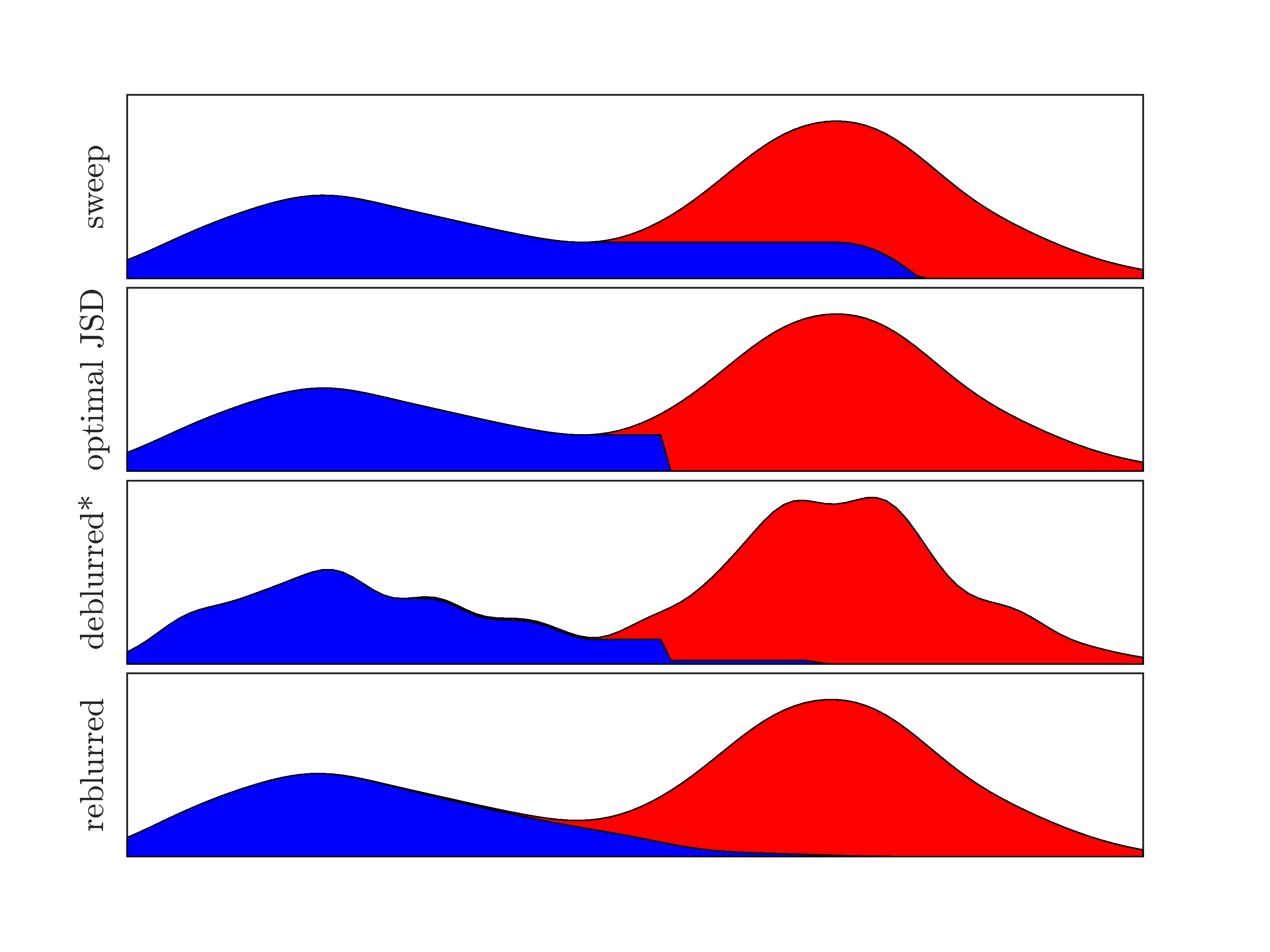}
\caption{ \label{fig:OldFaithfulArea} TME applied to $n = 272$ waiting times between eruptions of the Old Faithful geyser. Panels show area plots of unimodal decompositions obtained by (top) the ``sweep algorithm'' on a TDE with bandwidth given by \eqref{eq:TDE}; (second from top) the result of \eqref{eq:TME} on the same; (second from bottom) the ``deblurred'' result of \eqref{eq:TME} on a TDE with bandwidth given by \eqref{eq:TDEminimal}; (bottom) the result of ``reblurring'' by convolving the deblurred mixture with a Gaussian kernel with bandwidth given by \eqref{eq:Deltah}. Note that three of the four mixture estimates have the same density, but that the deblurred density is different (we highlight this with a ``*'' annotation).
} 
\end{figure} %

We present two phenomenologically illustrative examples. First, in Figures \ref{fig:OldFaithfulArea} and \ref{fig:OldFaithfulLine} we consider the $n = 272$ waiting times between eruptions of the Old Faithful geyser from the data set in \cite{hardle2012smoothing}. Then, in Figures \ref{fig:ColorIndicesArea} and \ref{fig:ColorIndicesLine} we consider the $n = 2107$ Sloan Digitial Sky Survey $g-r$ color indices accessed from the VizieR database \cite{ochsenbein2000vizier} at \url{http://cdsarc.u-strasbg.fr/viz-bin/Cat?J/ApJ/700/523} and discussed in \cite{an2009galactic}; the latter example is replicated in \cite{BAETME}.

\begin{figure}[htbp]
\includegraphics[trim = 5mm 10mm 10mm 10mm, clip, width=\columnwidth,keepaspectratio]{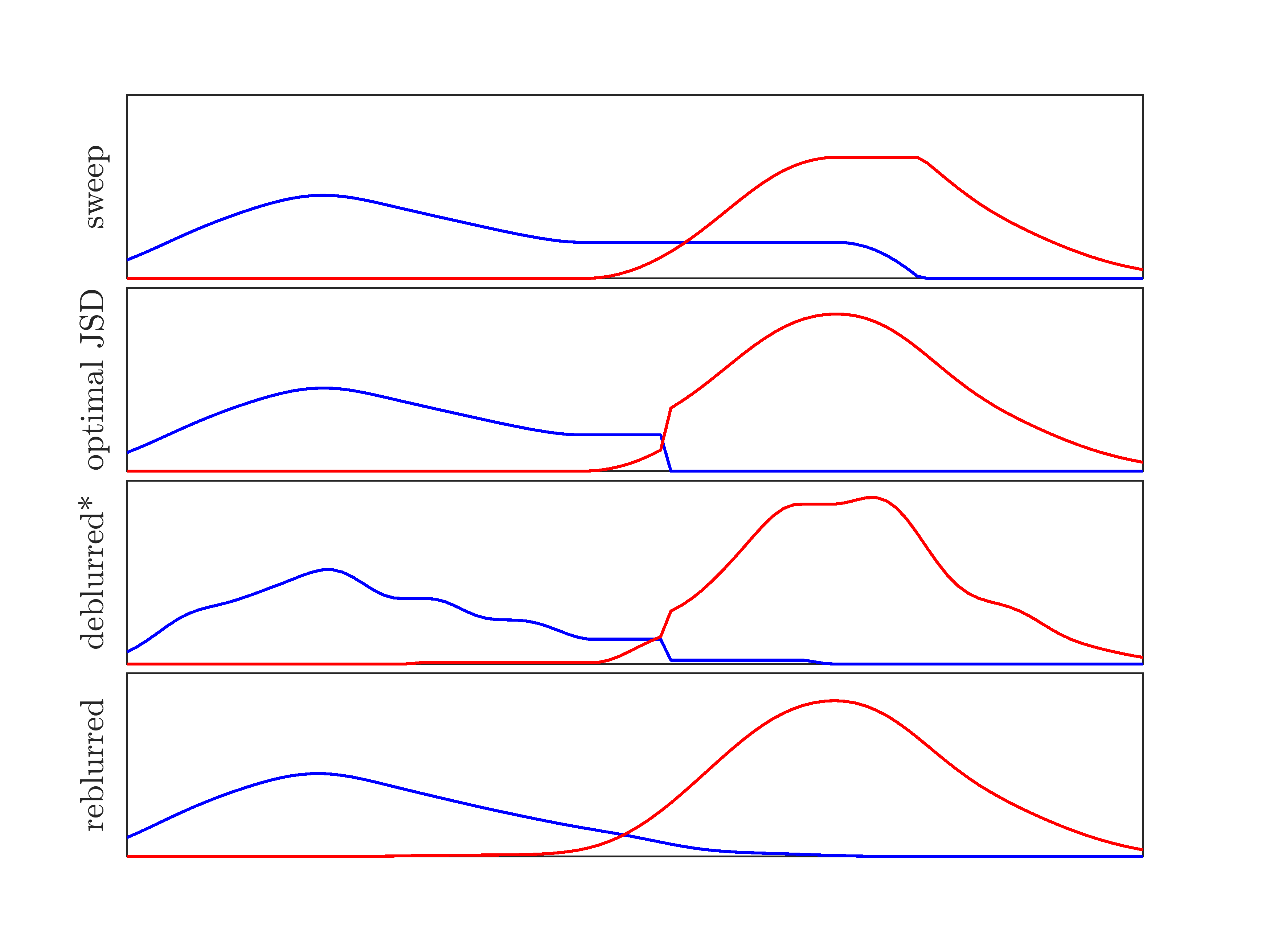}
\caption{ \label{fig:OldFaithfulLine} Line plots of the same decompositions as Figure \ref{fig:OldFaithfulArea}.
} 
\end{figure} %

\begin{figure}[htbp]
\includegraphics[trim = 5mm 10mm 10mm 10mm, clip, width=\columnwidth,keepaspectratio]{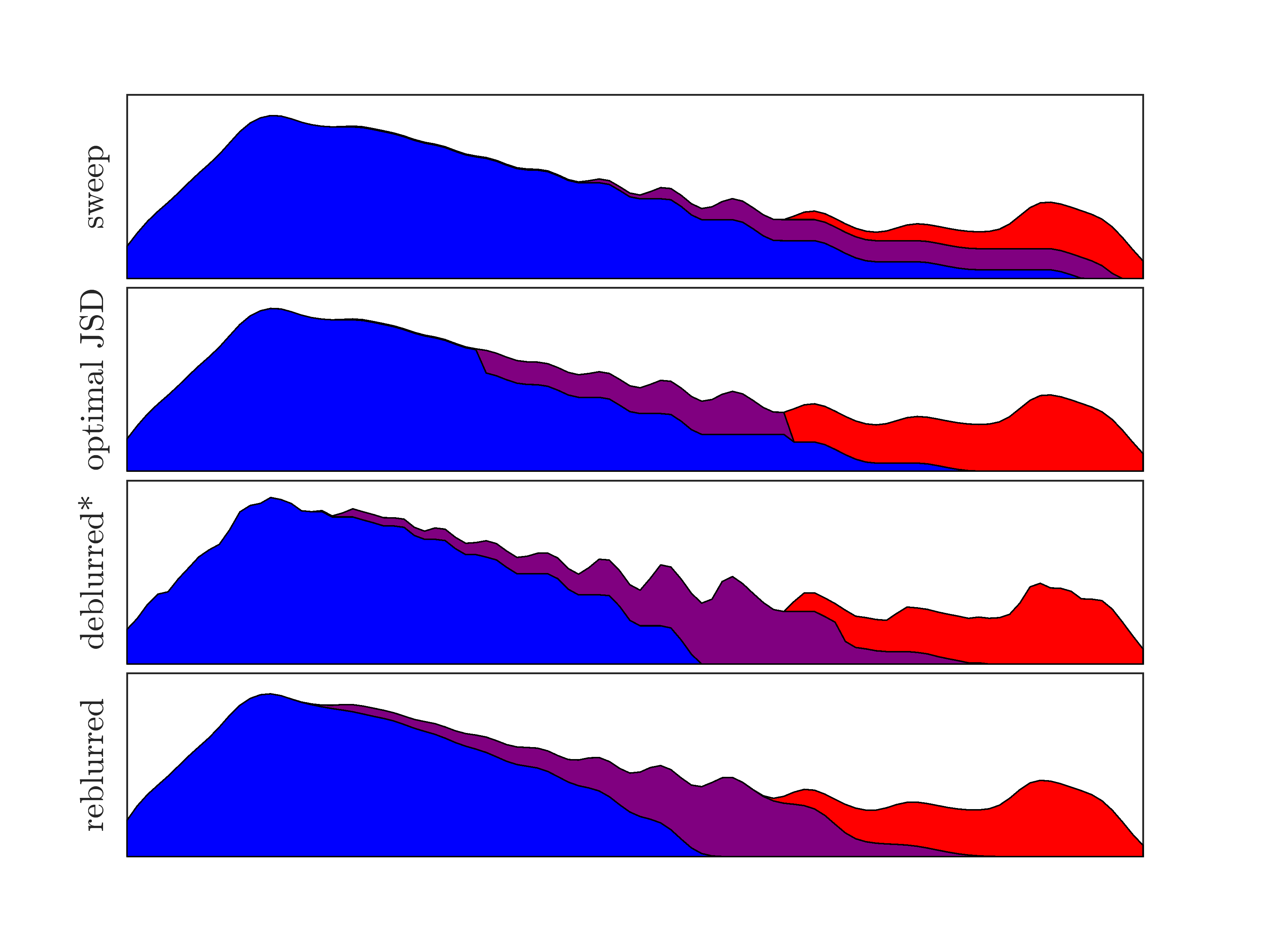}
\caption{ \label{fig:ColorIndicesArea} TME applied to $n = 2107$ $g-r$ color indices from \cite{an2009galactic}. Panels are otherwise as in Figure \ref{fig:OldFaithfulArea}.
} 
\end{figure} %

\begin{figure}[htbp]
\includegraphics[trim = 5mm 10mm 10mm 10mm, clip, width=\columnwidth,keepaspectratio]{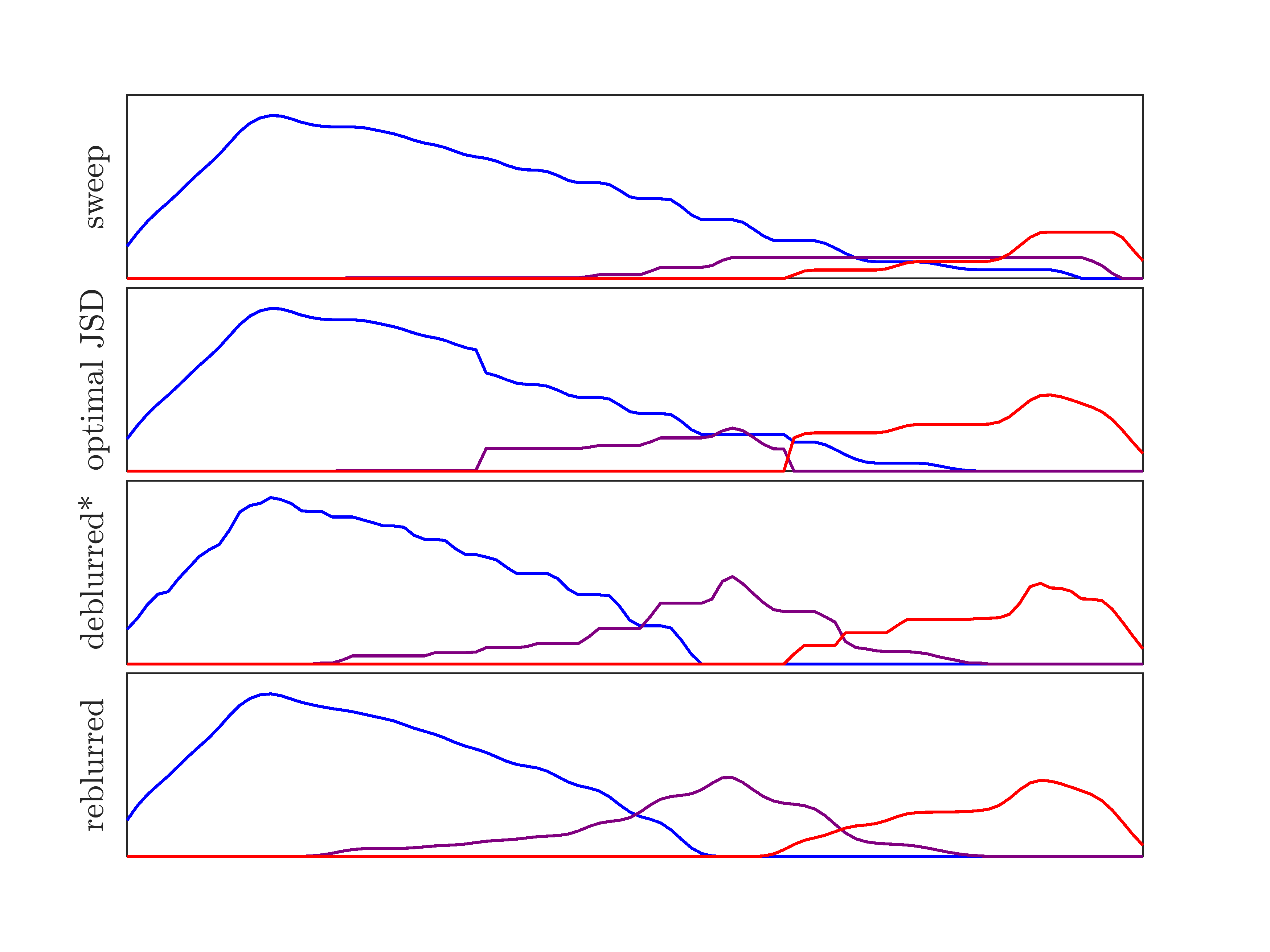}
\caption{ \label{fig:ColorIndicesLine} TME applied to $n = 2107$ $g-r$ color indices from \cite{an2009galactic}. Panels are otherwise as in Figure \ref{fig:OldFaithfulLine}.
} 
\end{figure} %

\begin{figure}[htbp]
\includegraphics[trim = 0mm 20mm 10mm 20mm, clip, width=\columnwidth,keepaspectratio]{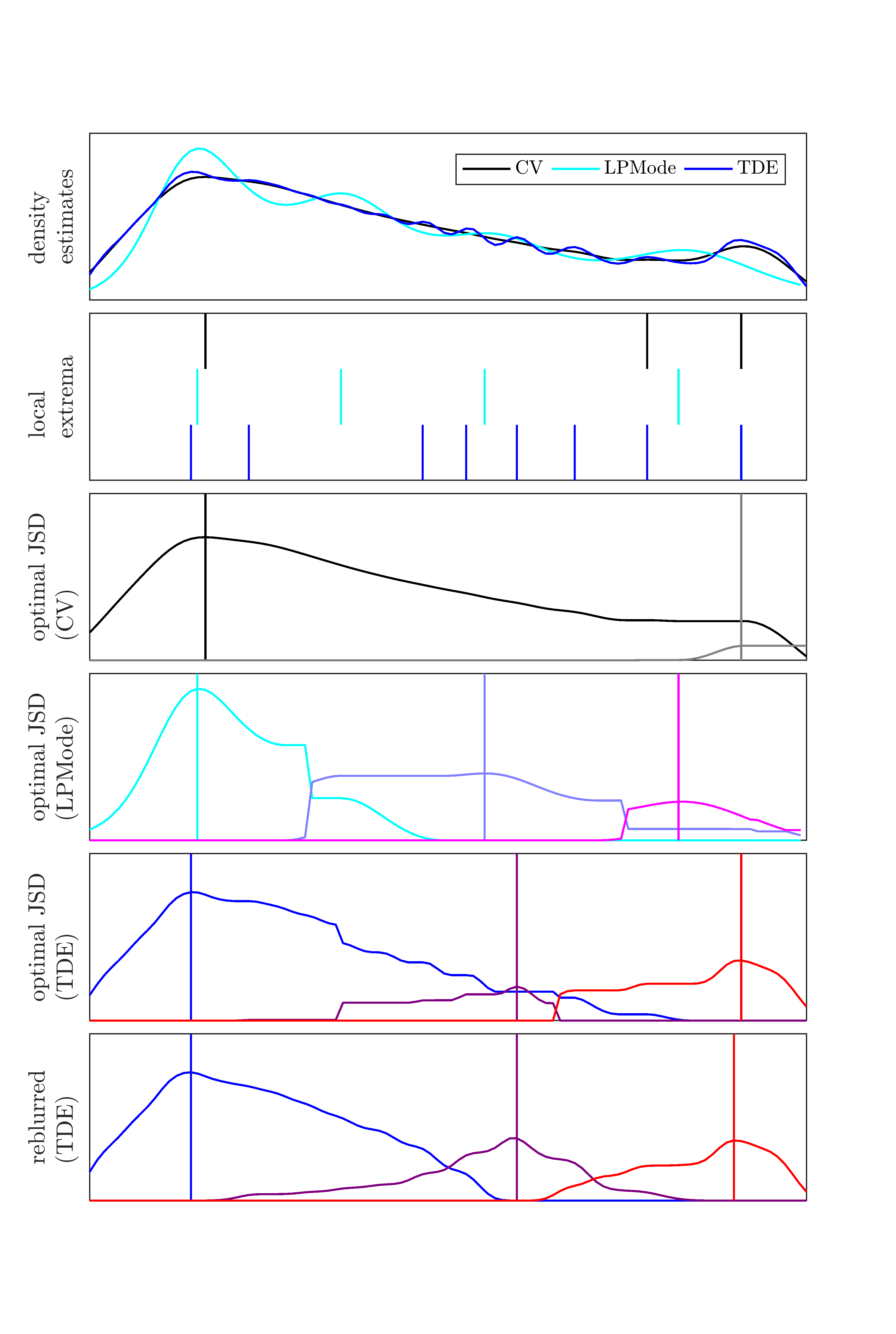}
\caption{ \label{fig:CvLpmodeTmeBumpHunting}Structural analyses of the color index data from Figure \ref{fig:ColorIndicesLine}. Top panel: density estimates from CV, the default LPMode algorithm of \cite{mukhopadhyay2017large} (kindly provided by its author), and TDE. Second panel from top: the local maxima of the density estimates above. Third through fifth panels from top: the result of \eqref{eq:TME} on CV, LPMode, and TDE, respectively, with maxima of components indicated. Bottom panel: the reblurred mixture from the bottom panel of Figure \ref{fig:ColorIndicesLine}, augmented with component maxima.}
\end{figure} 

As suggested, several phenomena are readily apparent from these examples. First, mixtures obtained via the sweep algorithm are manifestly parity-dependent, i.e., the direction of sweeping matters; second, mixtures obtained via TME alone exhibit artificial anti-overlapping behavior; third, deblurring followed by reblurring preserves unimodality, the overall density, a topologically persistent invariant (viz., the unimodal category) and the spirit of information-theoretical optimality while producing an obviously better behaved mixture; fourth and finally, the various techniques involved here can significantly shift classification/decision boundaries based on the dominance of various mixture components.

While the data in Figures \ref{fig:OldFaithfulArea} and \ref{fig:OldFaithfulLine} is at least qualitatively approximated by a two-component Gaussian mixture, it is clear that a three-component Gaussian mixture cannot capture the highly oscillatory behavior of the density in Figures \ref{fig:ColorIndicesArea} and \ref{fig:ColorIndicesLine}. Indeed, this example illustrates how such oscillatory behavior can actually arise from a unimodal mixture with many fewer components than might naively appear to be required.

Figure \ref{fig:CvLpmodeTmeBumpHunting} shows that there are strong and independent grounds to conclude that the color index data of Figure \ref{fig:ColorIndicesLine} is produced by a unimodal mixture of three components, with componentwise modes as suggested by TME, and furthermore that CV undersmooths this data. For each density estimate shown, each componentwise maximum of the corresponding mixture \eqref{eq:TME} is virtually identical to one of the local maxima of the density estimate itself: this is a consequence of the anti-overlapping tendency described above. 

In particular, the fourth panel of Figure \ref{fig:CvLpmodeTmeBumpHunting} illustrates that it is possible and potentially advantageous to use TME as an alternative mode identification technique in the LPMode algorithm of \cite{mukhopadhyay2017large}. Furthermore, while we have not implemented a reliable Fourier deconvolution/reblurring algorithm of the sort hinted at in \S \ref{sec:deblurring}, the fifth and sixth panels of Figure \ref{fig:CvLpmodeTmeBumpHunting} suggest that this is not particularly important for the narrowly defined task of mode finding/bump hunting.

\section{\label{sec:remarks}Remarks}

While the $O(M^4 N^3)$ arithmetic operations of TME as implemented in Algorithm \ref{alg:tme} might seem uncomfortably high, in practice $M$ is generally quite small and it is reasonable to enforce $N = 10^2$ as in fact we do in the MATLAB implementation \cite{BAETME} and \S \ref{sec:examples}. Furthermore, avoiding redundant perturbations and incrementally computing $J$ would dramatically reduce the computational complexity. Still, even without the benefit of any such refinements, the color index example in \S \ref{sec:examples} runs in less than 40 seconds. 

Note also that TDE and de/reblurring as respectively implemented in Algorithms \ref{alg:tde} and \ref{alg:blur} are relatively computationally inexpensive. The former case is helped by resampling data to on the order of $10^3$ quantiles, which has no material effect on the output in cases where kernels are appropriate to use. That said, using a fast generalized Gauss transform \cite{spivak2010fast} would dramatically accelerate TDE.

Even though TME is presently limited to one dimension, it is still very useful due to the preponderance of one-dimensional problems. For instance, TME is a good candidate to improve on some of the best practical unsupervised image thresholding techniques \cite{kapur1985thresholding,kittler1986thresholding,sezgin2004thresholding}, and has prospects for enhancing data/sensor fusion, information-theoretical analysis of time series (by using sliding time windows to define samples), and many other tasks. 

Furthermore, the one-dimensional framework can be used with random projections of high-dimensional data in a way that is likely to yield improvements for model selection \cite{feng2007pgmeans} (cf. \cite{kalai2012disentangling}) and anomaly detection \cite{pevny2016loda}. We plan to explore these topics in future work, with the associated intent of gauging the art of the possible with respect to determining unimodal decompositions in dimension $> 1$.

That said, the extension of TME to dimension $> 1$ will require effort and mathematical tools well beyond those used in this paper. One reason is that computing a unimodal decomposition is algorithmically undecidable in high dimensions, a property inherited from the problem of determining contractibility of simplicial complexes \cite{tancer2016collapsible} and the geometric realization theorem \cite{edelsbrunner2010topology} applied to level and upper excursion sets. Therefore, extending the constructions of this paper will require some modification of the notion of unimodal category in dimension $> 1$, approximations and/or heuristics. 

For example, restricting to convex versus contractible upper excursion sets is probably desirable on intuitive as well as computational grounds. However, even in two dimensions the corresponding problem of constructing minimal convex partitions of polygons is still NP-hard. On the other hand, the case without interior holes is efficiently solvable \cite{orourke2017polygons} and there is a quasi-polynomial time approximation scheme for the general case \cite{bandyapadhyay2015approximation}. The heuristic of \cite{liu2010convex} seems to be a good starting point for exploring relevant tradeoffs.

An appropriate tactic for directly leveraging the one-dimensional framework \emph{en route} to higher dimensions appears to be tomography in the spirit of the topological Radon transform \cite{ghrist2014elementary}. Besides working with random projections in this vein, another sensible approach (suggested to the author by Robert Ghrist) is to foliate \cite{lawson1970foliations} the domain of a sample (in practice this would just mean taking a family of parallel lines) and perform TME on data in tubular neighborhoods of nearby leaves of the foliation, then assemble the results, essentially by interpolating. Here topological tools such as sheaves and Morse theory seem inevitably to be required in order to do things in a globally coherent way. A suitable member of the class of metrics on mixtures introduced in \cite{liu2000distance} that provides data relevant for assembly as a byproduct will likely also be necessary.

Finally, we note that there are prospects for recursively coupling TME and TDE: the idea here is to pull mixture components back to weighted subsamples, then re-run TDE (or in the unimodal case, CV) on these individually. The resulting variable-bandwidth mixture estimator would give a multiresolution description of data that would be be likely to yield further improvements in many applications.

\appendix

\section{\label{sec:Convexity}Convexity}

The following lemma shows that $J$ is convex as we gradually shift part of one mixture component to another.

\begin{lemma}
\label{lem:convexity}
Let $|(\pi,p)| = 3$ and define
\begin{eqnarray}
\pi_{12,t} & := & \pi_1 + (1-t) \pi_2; \nonumber \\
\pi_{23,t} & := & t \pi_2 + \pi_3; \nonumber \\
p_{12,t} & := & \frac{\pi_1 p_1 + (1-t) \pi_2 p_2}{\pi_{12,t}}; \nonumber \\
p_{23,t} & := & \frac{t \pi_2 p_2 + \pi_3 p_3}{\pi_{23,t}}, \nonumber
\end{eqnarray}
so that $\left \langle \left ( \pi_{12,t}, \pi_{23,t} \right ), \left ( p_{12,t}, p_{23,t} \right ) \right \rangle = \langle \pi, p \rangle$.
The function $g_{\pi,p} : [0,1] \rightarrow [0,\infty)$ defined by
\begin{equation}
g_{\pi,p}(t) := J \left ( \left ( \pi_{12,t}, \pi_{23,t} \right ), \left ( p_{12,t}, p_{23,t} \right ) \right )
\end{equation}
satisfies
\begin{equation}
g_{\pi,p}(t) \le t \cdot g_{\pi,p}(1) + (1-t) \cdot g_{\pi,p}(0).
\end{equation}
\end{lemma}

\begin{proof}
We have that $g_{\pi,p}(t) = H(\langle \pi, p \rangle) - \left \langle \left ( \pi_{12,t}, \pi_{23,t} \right ), \left ( H(p_{12,t}), H(p_{23,t}) \right ) \right \rangle$.
Furthermore, if we write $\pi_{12} := \pi_1 + \pi_2$, $\pi_{23} := \pi_2 + \pi_3$, $p_{12} := \frac{\pi_1 p_1 + \pi_2 p_2}{\pi_{12}}$, and $p_{23} := \frac{\pi_2 p_2 + \pi_3 p_3}{\pi_{23}}$, then
\begin{eqnarray}
\pi_{12,t} & = & t \pi_1 + (1-t)\pi_{12}; \nonumber \\
\pi_{23,t} & = & t \pi_{23} + (1-t) \pi_3; \nonumber \\
p_{12,t} & = & \frac{t \pi_1 p_1 + (1-t) \pi_{12} p_{12}}{\pi_{12,t}}; \nonumber \\
p_{23,t} & = & \frac{t \pi_{23} p_{23} + (1-t) \pi_3 p_3}{\pi_{23,t}}. \nonumber
\end{eqnarray}

It is well known that $H$ is a concave functional: from this it follows that
\begin{eqnarray}
H(p_{12,t}) & \ge & \frac{t\pi_1}{\pi_{12,t}} H(p_1) + \frac{(1-t)\pi_{12}}{\pi_{12,t}} H(p_{12}); \nonumber \\
H(p_{23,t}) & \ge & \frac{t \pi_{23}}{\pi_{23,t}} H(p_{23}) + \frac{(1-t) \pi_3}{\pi_{23,t}} H(p_3). \nonumber
\end{eqnarray}
Therefore
\begin{eqnarray}
g_{\pi,p}(t) & = & H(\langle \pi, p \rangle) \nonumber \\
& & - \left \langle \left ( \pi_{12,t}, \pi_{23,t} \right ), \left ( H(p_{12,t}), H(p_{23,t}) \right ) \right \rangle \nonumber \\
& \le & H(\langle \pi, p \rangle) - t\pi_1 H(p_1) - (1-t)\pi_{12} H(p_{12}) \nonumber \\
& & - t \pi_{23} H(p_{23}) - (1-t) \pi_3 H(p_3) \nonumber \\
& = & t H(\langle \pi, p \rangle) - t \langle (\pi_1,\pi_{23}), (H(p_1), H(p_{23})) \rangle \nonumber \\
& & + (1-t) \cdot H(\langle \pi, p \rangle) \nonumber \\
& & - (1-t) \cdot \langle (\pi_{12},\pi_3), (H(p_{12}), H(p_3)) \rangle \nonumber \\ 
& = & t \cdot g_{\pi,p}(1) + (1-t) \cdot g_{\pi,p}(0) \nonumber
\end{eqnarray}
as claimed.
\end{proof}

\section{\label{sec:PreservingUnimodality}Preserving Unimodality}

Suppose that $(\pi, p)$ is a unimodal mixture on $\mathbb{R}$ with $|(\pi, p)| > 1$. We would like to determine how we can perturb two components of this mixture so that the result is still unimodal and yields the same density. In the event that the mixture is piecewise affine and continuous (or piecewise constant) the space of permissible perturbations can be characterized by the following 

\begin{lemma} 
\label{lem:unimodality}
For $0 \le k \le N$, let $y_k \in [0, \infty)$ be such that $y_0 = 0 = y_N$ and there are integers $\ell, u$ satisfying $0 < \ell \le u < N$ with
\begin{equation}
\label{eq:nnUnimodal}
y_0 \le \dots \le y_{\ell-1} < y_\ell = \dots = y_u > y_{u+1} \ge \dots \ge y_N.
\end{equation} 
(That is, $y_1,\dots,y_{N-1}$ is a nonnegative, nontrivial unimodal sequence.) Then for $1 \le r \le N-1$ and $\varepsilon^-_r \ge 0$, $y_k - \delta_{kr} \varepsilon^-_r$ is nonnegative and unimodal iff
\begin{equation}
\label{eq:give}
\varepsilon^-_r \le y_r - \min \{y_{r-1}, y_{r+1}\}.
\end{equation}
Similarly, for $\varepsilon^+_r \ge 0$, $y_k + \delta_{kr} \varepsilon^+_r$ is nonnegative and unimodal iff
\begin{equation}
\label{eq:take}
\varepsilon^+_r \le \begin{cases} \infty & \text{if } \ell-1 \le k \le u+1 \\ \max \{y_{r-1}, y_{r+1}\} -y_r & \text{otherwise.} \end{cases}
\end{equation}
\end{lemma}

\begin{proof}[Proof (sketch).]
We first sketch (\eqref{eq:give},$\Leftarrow$). Nonegativity follows from $0 \le \min \{y_{r-1}, y_{r+1} \} \le y_r - \varepsilon^-_r$. Unimodality follows from a series of trivial checks for the cases $1 \le r < \ell-1$, $r = \ell-1$, $r = \ell$, and $\ell < r < u$: the remaining cases $r = u$, $r = u+1$, and $u+1 < r \le N-1$ follow from symmetry. For example, in the case $1 \le r < \ell-1$, we only need to show that $y_{r-1} \le y_r - \varepsilon^-_r \le y_{r+1}$. 

A sketch of (\eqref{eq:give},$\Rightarrow$) amounts to using the same cases and symmetry argument to perform equally trivial checks. For example, in the case $1 \le r < \ell-1$, we have $\varepsilon^-_r \le y_r - y_{r-1} \le y_r - \min \{ y_{r-1}, y_{r+1}\}$.

The proof of \eqref{eq:take} is mostly similar to that of \eqref{eq:give}: the key difference here is that any argument adjacent to or at a point where the maximum is attained can have its value increased arbitrarily without affecting unimodality (or nonnegativity).
\end{proof}

The example in Figure \ref{fig:unimodalGiveTake} is probably more illuminating than filling in the details of the proof sketch above.

\section*{Acknowledgements}

The author thanks Adelchi Azzalini, Robert Ghrist, Subhadeep Mukhopadhyay, and John Nolan for their helpful and patient discussions, and reviewers for their comments and advice. 

This material is based upon work supported by the Defense Advanced Research Projects Agency (DARPA) and the Air Force Research Laboratory (AFRL). Any opinions, findings and conclusions or recommendations expressed in this material are those of the author(s) and do not necessarily reflect the views of DARPA or AFRL.

\bibliography{tmeICMLbib}
\bibliographystyle{icml2018}

\end{document}